\documentclass[11pt,a4paper]{article}
\usepackage{latexsym,amsfonts,amsthm,amsmath,amssymb,cases}
\usepackage{cite}
\usepackage{enumerate}
\usepackage{graphicx}
\usepackage{color}
\usepackage{bm}

\numberwithin{equation}{section}
\newtheorem{thm}{Theorem}[section]

\newtheorem{lemma}{Lemma}[section]

\newtheorem{remark}{Remark}[section]

\newtheorem{defn}{Definition}[section]

\newcommand{\ds}{\displaystyle}
\def\mathrm{\mbox}

\numberwithin{remark}{section}

\begin{document}
\title{{\Large \bf Asymptotic behavior for a finitely degenerate semilinear pseudo-parabolic equation}\thanks{This work was supported by the National Natural Science Foundation of China (12171339).}}
\author{Xiang-kun Shao$^a$, Xue-song Li$^a$, Nan-jing Huang$^a $\footnote{Corresponding author.  E-mail addresses: nanjinghuang@hotmail.com; njhuang@scu.edu.cn} and Donal O'Regan$^b$\\
{\small\it a. Department of Mathematics, Sichuan University, Chengdu, Sichuan 610064, P.R. China}\\
{\small\it b. School of Mathematical and  Statistical Sciences,  University of Galway, Ireland}}
\date{}
\maketitle
\begin{center}
\begin{minipage}{5.5in}
\noindent{\bf Abstract.} This paper investigates the initial boundary value problem of a finitely degenerate semilinear pseudo-parabolic equation associated with H\"{o}rmander's operator. Based on the global existence of solutions in previous literature, the exponential decay estimate of the energy functional is obtained. Moreover, by developing some novel estimates about solutions and using the energy method, the upper bounds of both blow-up time and blow-up rate and the exponential growth estimate of blow-up solutions are determined. In addition, the lower bound of blow-up rate is estimated when a finite time blow-up occurs. Finally, it is established that as time approaches infinity, the global solutions strongly converge to the solution of the corresponding stationary problem. These results complement and improve the ones obtained in the previous literature.
\end{minipage}
\end{center}

\section{Introduction}
\quad\quad We discuss the initial-boundary value problem (IBVP) of the finitely degenerate semilinear pseudo-parabolic equation:
\begin{align}\label{1.1}
 \left\{\begin{array}{ll}
     \ds \phi_{t}-\Delta_X\phi-\Delta_X\phi_t=|\phi|^{p-1}\phi,\ \ \ &x\in\Omega,\ t>0,\\
     \ds \phi(x,t)=0,&x\in\partial\Omega,\ t>0,\\
     \ds \phi(x,0)=\phi_0(x),\ \ &x\in\Omega,
    \end{array}\right.
\end{align}
where $X=(X_1,X_2,\cdots,X_{m-1},X_m)$ is a system of $C^\infty$ smooth vector fields defined on an open domain $\tilde{\Omega}\subset\mathbb{R}^n$ $(n\geq2)$, $\Omega\subset\subset \tilde{\Omega}$ is open and bounded, every $X_i$ is a formally skew-adjoint operator, i.e. $X_i^*=-X_i$, the operator $\Delta_X$ is defined by
\begin{equation*}
\Delta_X:=\sum\limits_{i=1}\limits^{m}X_i^2,
\end{equation*}
and $\phi_0(x)\in H_{X,0}^1(\Omega)$ where $H_{X,0}^1(\Omega)$ is defined below. Furthermore, we assume that
\begin{description}
  \item{$(A_1)$} $\partial\Omega$ is $C^\infty$ smooth and non-characteristic for $X$;
  \item{$(A_2)$} $X$ satisfies the H\"{o}rmander's condition on $\tilde{\Omega}$ (see \cite{MR222474}):
      For each point $x\in\tilde{\Omega}$, there exists an integer $Z>0$ such that $X_1,\cdots,X_m$ and all their commutators with length up to $k\leq Z$ span the tangent space. Here, the $k$-th order commutators of $X_1,\cdots,X_m$ are defined as
      \begin{equation*}
      X_I=[X_{i_1},[X_{i_2},\cdots[X_{i_{k-1}},X_{i_k}]\cdots]],\quad 1\leq i_j\leq m,
      \end{equation*}
      where $I=(i_1,i_2,\cdots,i_k)$, and $k=|I|$ denotes the length of the commutator. The minimal positive integer $Z$ satisfying the H\"{o}rmander's condition is called the H\"{o}rmander's index of $\tilde{\Omega}$ related to $X$;
  \item{$(A_3)$} The parameters $p$ satisfies
\begin{equation*}
p\in\left(1,\frac{r+2}{r-2}\right),
\end{equation*}
where $r>2$ is the generalized M\'{e}tivier index of $\Omega$ related to $X$, defined as follows (see \cite{MR427858,MR3412394,MR3354435}): Let $(A_2)$ hold and $Z$ be the H\"{o}rmander's index. For each point $x\in\tilde{\Omega}$, denote by $V_l(x)$~$(1\leq l\leq Z)$ the subspace of the tangent space at $x$ spanned by all commutators of $X_1,\cdots,X_m$ with length at most $l$, and let $v_l(x)$ be its dimension. Define \begin{equation*}
v(x):=\sum_{l=1}^Zl(v_l(x)-v_{l-1}(x)),\hbox{ where }v_0(x):=0.
\end{equation*}
For $\Omega\subset\subset \tilde{\Omega}$, the generalized M\'{e}tivier index of $\Omega$ related to $X$ is defined as
\begin{equation*}
r:=\max_{x\in\overline{\Omega}}v(x).
\end{equation*}
\end{description}

When $X=(\partial_{x_1},\partial_{x_2},\cdots,\partial_{x_n})$ (here, the H\"{o}rmander's index $Z=1$ and the generalized M\'{e}tivier index $r=n$), $\Delta_X$ is reduced to the standard Laplacian operator $\Delta$.
Now we recall some facts about the following classical pseudo-parabolic equation
\begin{equation}\label{1.2}
\phi_t-\Delta\phi-\Delta\phi_t=f(\phi),
\end{equation}
which can be used to describe the unsaturated flow in porous media with dynamic capillary pressure \cite{MR2593053}, the diffusion of substance through solids or gases \cite{MR586062,MR1575106}, the infiltration of a uniform liquid through fissured rocks \cite{Barenblatt1960theory}, the non-stationary processes in semiconductors in the presence of sources \cite{MR2198478,MR2290551}, the heat conduction involving two temperatures \cite{chen1968theory}, and so on.
In recent decades, research on the IBVP and the Cauchy problem of \eqref{1.2} has attracted widespread attention from scholars. In particular, \eqref{1.2} with a polynomial source $f$ being $\phi^q$ or $|\phi|^{q-1}\phi$ has aroused much  interest  and existence, uniqueness, asymptotic behavior and blow-up properties of solutions was considered  in the literature \cite{MR2523294,MR2981259,MR3045640,MR3372307,MR3745341}. For more results on pseudo-parabolic equation, we refer the reader to  \cite{MR437936,MR330774,MR264231,MR3987484,MR2454799,MR4145829,MR4295069,Cao2024review,MR4839772} and the references therein.

Returning to our consideration, note  H\"{o}rmander \cite{MR222474} in 1967 firstly introduced the concepts of H\"{o}rmander's condition, finite degenerate vector field and H\"{o}rmander's operator, and proved the hypoellipticity of H\"{o}rmander's operator. Since the vector fields $X$ in \eqref{1.1} satisfies the H\"{o}rmander's condition i.e. assumption $(A_2)$, one says that $X$ is finite degenerate and H\"{o}rmander's operator $\Delta_X$ is a finitely degenerate elliptic operator. The relevant physical background and mathematical properties of such operators can be found in \cite{MR427858,MR512213,MR762360,MR865430,MR3154431,MR2200233,MR3412394,MR222474,MR3354435}. A large number of researchers have studied PDEs involving finitely degenerate elliptic operators and significant progress was made, see, for example, \cite{MR1087376,MR1305697,MR2200233,MR222474,MR3961341,MR4068818,MR4709392}, etc, and it is now still an active field of research. It is worth noting that the classical non-degenerate equation describes the isotropic process in natural phenomena, which is quite idealized, while the degenerate equation often reveals the characteristic of anisotropy in the studied substance, which is sometimes more in line with reality. Therefore, the finitely degenerate semilinear pseudo-parabolic equation \eqref{1.1} may be applied to some physical phenomena with anisotropic processes and is worth studying.

Recently, Chen and Xu \cite{MR3961341} used a series of potential wells as tools to prove the existence of global and blow-up solutions of \eqref{1.1} under subcritical and critical initial energies, and estimated both the exponential decay of global solutions and a lower bound of blow-up time for blow-up solutions. Subsequently, Liu et al. \cite{MR4854159} proved the global existence and finite time blow-up of solutions to problem \eqref{1.1} under supercritical initial energy, established the existence of blow-up solutions for arbitrary initial energy levels, and proved that the global solutions of problem \eqref{1.1} weakly converge to the solution of the corresponding stationary problem as time approaches infinity. Moreover, Liu and Tian \cite{MR4709392} established the well-posed of solutions and the existence of global attractors for problem \eqref{1.1} with a general source term.

Although significant advances related to problem \eqref{1.1} have been achieved in previous studies, the asymptotic behavior of solutions to \eqref{1.1} still requires further in-depth analysis, which is of crucial importance for understanding the evolution mechanisms of pseudo-parabolic systems. In order to introduce the main research objective of this paper, let us firstly review the relevant spaces, sets and notations used in the paper. Recall the weighted Sobolev space associated with the vector fields $X$ (see \cite{MR222474})
\begin{equation*}
H_X^1(\tilde{\Omega})=\left\{\varphi\in L^2(\tilde{\Omega})\mid X_i\varphi\in L^2(\tilde{\Omega}),\quad i=1,2,\cdots,m\right\},
\end{equation*}
which is a Hilbert space with the norm
\begin{equation*}
\|\varphi\|_{H_X^1(\tilde{\Omega})}^2=\|X\varphi\|_{L^2(\tilde{\Omega})}^2+\|\varphi\|_{L^2(\tilde{\Omega})}^2,\quad\forall\varphi\in H_X^1(\tilde{\Omega}),
\end{equation*}
where
\begin{equation*}
\|X\varphi\|_{L^2(\tilde{\Omega})}^2=\sum_{i=1}^m\|X_i\varphi\|_{L^2(\tilde{\Omega})}^2.
\end{equation*}
In what follows, we denote the $L^2(\Omega)$-inner product as $(\cdot,\cdot)$ and the $L^q(\Omega)$-norm ($1\leq q\leq\infty$) as $\|\cdot\|_q$. The closure of $C_0^\infty(\Omega)$ in $H_X^1(\tilde{\Omega})$ is represented as the Hilbert space $H_{X,0}^1(\Omega)$ that is equipped with the norm
\begin{equation}\label{Hfanshu}
\|\varphi\|_{H_{X,0}^1(\Omega)}^2=\|X\varphi\|_2^2+\|\varphi\|_2^2,\quad\forall\varphi\in H_{X,0}^1(\Omega).
\end{equation}
Let $\lambda_1>0$ be the first Dirichlet eigenvalue of $-\Delta_X$. If $(A_1)$-$(A_2)$ are satisfied, then Proposition 2.3 of \cite{MR3961341} states that
\begin{equation}\label{tezhengzhi}
\lambda_1\|\varphi\|_2^2\leq\|X\varphi\|_2^2,\ \ \ \forall\varphi\in H_{X,0}^1(\Omega).
\end{equation}
From Proposition 2.4 of \cite{MR4854159}, it follows that if $(A_1)$-$(A_3)$ hold, then the embedding
$H_{X,0}^1(\Omega)\hookrightarrow L^{p+1}(\Omega)$ is compact. Thus we obtain
\begin{equation}\label{x1111-33}
\|\varphi\|_{p+1}\leq C\|\varphi\|_{H_{X,0}^1(\Omega)},\ \ \ \forall\varphi\in H_{X,0}^1(\Omega),
\end{equation}
where $C>0$ denotes the optimal embedding constant.

For all $\varphi\in H_{X,0}^1(\Omega)$, the energy functional $J$ and the Nehari functional $I$ are defined by
\begin{align}
J(\varphi):=\frac{1}{2}\|X\varphi\|_2^2-\frac{1}{p+1}\|\varphi\|_{p+1}^{p+1}\label{J}
\end{align}
and
\begin{align}
I(\varphi):=\|X\varphi\|_2^2-\|\varphi\|_{p+1}^{p+1}.\label{I}
\end{align}
Clearly,
\begin{equation}\label{JI}
J(\varphi)=\frac{1}{p+1}I(\varphi)+\frac{p-1}{2(p+1)}\|X\varphi\|_2^2.
\end{equation}
Let
\begin{equation}\label{d}
d:=\inf_{\varphi\in \mathcal{N}}J(\varphi)
\end{equation}
denote the mountain pass level, where
\begin{equation}\label{N}
\mathcal{N}=\left\{\varphi\in H_{X,0}^1(\Omega)\setminus\{0\}\mid I(\varphi)=0\right\}
\end{equation}
is the Nehari manifold. By Lemma 3.1 in \cite{MR3961341}, one has
\begin{equation}\label{djingquezhi}
d=\frac{p-1}{2(p+1)}C_*^{-\frac{2(p+1)}{p-1}}
\end{equation}
with
\begin{equation}\label{qianru}
C_*:=\sup_{\varphi\in H_{X,0}^1(\Omega)\setminus\{0\}}\frac{\|\varphi\|_{p+1}}{\|X\varphi\|_2}.
\end{equation}
Moreover, we set
\begin{equation}\label{Nzhengfu}
\mathcal{N}_+:=\left\{\varphi\in H_{X,0}^1(\Omega)\mid I(\phi)>0\right\}\cup\{0\},\quad \mathcal{N}_-:=\left\{\varphi\in H_{X,0}^1(\Omega)\mid I(\phi)<0\right\}.
\end{equation}

The weak solution of \eqref{1.1} is defined as follows.
\begin{defn}\label{weaksolution}
Let $\phi_0\in H_{X,0}^1(\Omega)$ and $T>0$. A function $\phi\in L^\infty\left(0,T;H_{X,0}^1(\Omega)\right)$ with $\phi_t\in L^2(0,T;H_{X,0}^1(\Omega))$ is called a weak solution of \eqref{1.1} over $\Omega\times[0,T)$, if $\phi(x,0)=\phi_0(x)$ and the following equality holds:
\begin{equation}\label{weakdengshi}
(\phi_t,\psi)+(X\phi,X\psi)+(X\phi_t,X\psi)=(|\phi|^{p-1}\phi,\psi),\quad  \forall\psi\in H_{X,0}^1(\Omega),~t\in(0,T).
\end{equation}
\end{defn}

Let $\phi=\phi(t)$ be a weak solution to \eqref{1.1} with $t\in[0,T)$, where $T$ is the maximal existence time. Then according to \cite{MR3961341}, we have
\begin{equation}\label{nengliangdengshi}
J(\phi)+\int_0^t\|\phi_\tau\|_{H_{X,0}^1(\Omega)}^2d\tau=J(\phi_0),\quad \forall t\in[0,T)
\end{equation}
and
\begin{equation}\label{Idengshi}
\frac{d}{dt}\|\phi\|_{H_{X,0}^1(\Omega)}^2=-2I(\phi),\quad \forall t\in[0,T).
\end{equation}


In view of above, several natural questions arise:
\begin{description}
  \item{$(Q_1)$} According to \eqref{nengliangdengshi}, one can see that the energy functional $J(\phi)$ is decreasing with respect to $t$. However, can one estimate the decay rate of $J(\phi)$?
  \item{$(Q_2)$} Can one derive the upper bound of blow-up time in the case of $0\leq J(\phi_0)<d$? Even more interestingly, can one estimate the bounds of blow-up rate and the growth rate of blow-up solutions?
  \item{$(Q_3)$} Theorem 1.5 in \cite{MR4854159} establishes that as time tends to infinity, any global solution of problem \eqref{1.1} weakly converges to the solution of stationary problem associated with \eqref{1.1}. Can this result be further improved?
\end{description}

The principal objective of this paper is to address questions $(Q_1)$-$(Q_3)$.
As for $(Q_1)$, by constructing an appropriate auxiliary functional, we establish the exponential decay estimate of the energy functional.
As for $(Q_2)$, we note that although the upper bound of blow-up time under negative initial energy was derived via the energy method in \cite{MR4854159}, such an approach may not be directly extended to analyze blow-up phenomena for \eqref{1.1} with non-negative initial energy. In our analysis, we develop some novel estimates about solutions to \eqref{1.1} (i.e., Lemma \ref{yinliqi}) and combine them with the energy method to determine the upper bounds of blow-up time and rate under non-negative initial energy, while proving that blow-up solutions exhibit exponential growth patterns. Moreover, we obtain the lower bound estimate of blow-up rate when a finite-time blow-up occurs.
As for $(Q_3)$, through systematic application of differential-integral inequality techniques, we establish that as time approaches infinity, the global solution of \eqref{1.1} strongly converges to the solution of the stationary problem corresponding to \eqref{1.1}, thereby improving upon the weak convergence result in \cite{MR4854159}.

The structure of this paper is as follows. Section 2 is divided into subsections 2.1-2.3, which respectively state the main results obtained in this paper, containing the theorems, their proofs and some lemmas required for proofs.

\section{Main results}
\quad\quad
In this section, we focus on the questions $(Q_1)$-$(Q_3)$ above for \eqref{1.1} and applying the potential well method, the energy method and a series of differential integral inequality techniques to study the asymptotic behavior of global solutions and blow-up solutions to \eqref{1.1}.

\subsection{Decay estimate of the energy functional}
\quad\quad The following result concerns the exponential decay property of the energy functional, which is our answer to question $(Q_1)$.

\begin{thm}\label{decay}
Suppose that $(A_1)$-$(A_3)$ hold. Let $\phi=\phi(t)$, $t\in[0,T)$ be a weak solution to \eqref{1.1}. If $J(\phi_0)<d$ and $I(\phi_0)>0$, then $\phi$ exists globally and $\phi\in\mathcal{N}_+$ for $t\in[0,\infty)$. Moreover,  for $t\in[0,\infty)$,
\begin{equation*}
J(\phi)\leq\left(\|\phi_0\|_{H_{X,0}^1(\Omega)}^2+J(\phi_0)\right)e^{-\tilde{C}t},
\end{equation*}
where
\begin{equation*}
\tilde{C}=\frac{\varepsilon\lambda_1(p-1)}{2(\lambda_1+1)(p+1)+\lambda_1(p-1)}>0.
\end{equation*}
Here
\begin{equation}\label{vpsl}
\varepsilon=\frac{4(p+1)}{2+(p-1)\left[1-\left(\frac{J(\phi_0)}{d}\right)^{\frac{p-1}{2}}\right]^{-1}}>0
\end{equation}
and $\lambda_1$ is given in \eqref{tezhengzhi}.
\end{thm}
\begin{proof}
Since $J(\phi_0)<d$ and $I(\phi_0)>0$, we obtain from \cite[Theorem 1.4]{MR3961341} that the problem \eqref{1.1} admits a global solution $\phi$ and $\phi\in \mathcal{N}_+$ for $t\in[0,\infty)$. Then we have $I(\phi)\geq0$ for $t\in[0,\infty)$ by \eqref{Nzhengfu}. It is clear from \eqref{nengliangdengshi} and \eqref{JI} that
\begin{equation}\label{33}
J(\phi_0)\geq J(\phi)\geq\frac{p-1}{2(p+1)}\|X\phi\|_2^2,\ \ \ t\geq0.
\end{equation}
By \eqref{qianru} and \eqref{33}, we have
\begin{equation*}
\|\phi\|_{p+1}\leq C_*\|X\phi\|_2\leq C_*\sqrt{\frac{2J(\phi_0)(p+1)}{p-1}},
\end{equation*}
which, along with \eqref{qianru} and \eqref{djingquezhi}, yields
\begin{equation}\label{36}
\begin{split}
\|\phi\|_{p+1}^{p+1}&\leq C_*^2\|\phi\|_{p+1}^{p-1}\|X\phi\|_2^2\\
&\leq C_*^{p+1}\left[\frac{2J(\phi_0)(p+1)}{p-1}\right]^{\frac{p-1}{2}}\|X\phi\|_2^2\\
&=\left(\frac{J(\phi_0)}{d}\right)^{\frac{p-1}{2}}\|X\phi\|_2^2.
\end{split}\end{equation}
Then by \eqref{I} and \eqref{36}, one has
\begin{equation}\label{35}
I(\phi)=\|X\phi\|_2^2-\|\phi\|_{p+1}^{p+1}\geq\left[1-\left(\frac{J(\phi_0)}{d}\right)^{\frac{p-1}{2}}\right]\|X\phi\|_2^2.
\end{equation}
Set
\begin{equation*}
\mathcal{F}(t):=\|\phi\|_{H_{X,0}^1(\Omega)}^2+J(\phi),\quad t\geq0.
\end{equation*}
Then we deduce from \eqref{Hfanshu}, \eqref{tezhengzhi} and \eqref{33} that
\begin{equation}\label{38}
\mathcal{F}(t)\leq\frac{\lambda_1+1}{\lambda_1}\|X\phi\|_2^2+J(\phi)
\leq\left[\frac{2(\lambda_1+1)(p+1)}{\lambda_1(p-1)}+1\right]J(\phi).
\end{equation}
Moreover, by \eqref{nengliangdengshi}, \eqref{Idengshi}, \eqref{JI} and \eqref{35}, we obtain, for any constant $\varepsilon>0$,
\begin{equation}\label{37}
\begin{split}
\mathcal{F}'(t)&=-2I(\phi)-\|\phi_t\|_{H_{X,0}^1(\Omega)}^2\\
&\leq-2I(\phi)-\varepsilon J(\phi)+\frac{\varepsilon}{p+1}I(\phi)+\frac{\varepsilon(p-1)}{2(p+1)}\|X\phi\|_2^2\\
&\leq\left\{\frac{\varepsilon}{p+1}+\frac{\varepsilon(p-1)}{2(p+1)}\left[1-\left(\frac{J(\phi_0)}{d}\right)^{\frac{p-1}{2}}\right]^{-1}-2\right\}I(\phi)
-\varepsilon J(\phi).
\end{split}\end{equation}
Let $\varepsilon$ be as in  \eqref{vpsl}. Then we conclude from \eqref{38} and \eqref{37} that
\begin{equation*}
\mathcal{F}'(t)\leq-\varepsilon J(\phi)
\leq -\frac{\varepsilon\lambda_1(p-1)}{2(\lambda_1+1)(p+1)+\lambda_1(p-1)}\mathcal{F}(t)=:-\tilde{C}\mathcal{F}(t).
\end{equation*}
Therefore
\begin{equation*}
\mathcal{F}(t)\leq \mathcal{F}(0)e^{-\tilde{C}t},
\end{equation*}
which gives
\begin{equation*}
J(\phi)\leq \left(\|\phi_0\|_{H_{X,0}^1(\Omega)}^2+J(\phi_0)\right)e^{-\tilde{C}t}.
\end{equation*}
This ends the proof.
\end{proof}

\begin{remark}
By $J(\phi_0)<d$, $I(\phi_0)>0$ and \eqref{JI}, we obtain $0<J(\phi_0)<d$. Thus, the decay estimate in Theorem \ref{decay} makes sense.
\end{remark}

\subsection{Bounds of blow-up time and rate}
\quad\quad In this subsection, the upper bounds of blow-up time and blow-up rate as well as growth rate of blow-up solutions to \eqref{1.1} with subcritical initial energy are estimated, and the lower bound of blow-up rate is presented when a blow-up occurs, which solve question $(Q_2)$.

\begin{thm}\label{baopo}
Suppose that $(A_1)$-$(A_3)$ hold. Let $\phi=\phi(t)$, $t\in[0,T)$ be a weak solution to \eqref{1.1}. If $J(\phi_0)<d$ and $I(\phi_0)<0$, then $\phi$ blows up at some finite time $T$, namely,
\begin{equation*}
\lim_{t\rightarrow T^-}\|\phi\|_{H_{X,0}^1(\Omega)}^2=\infty.
\end{equation*}
Furthermore,
\begin{description}
  \item{$(i)$} the upper bounds of blow-up time and rate are given by
  \begin{equation}\label{ut}
  T\leq\frac{\|u_0\|_{H_{X,0}^1(\Omega)}^2}{C_1(C_1-2)(d-J(\phi_0))},\hbox{~~~if~~}0\leq J(\phi_0)<d,
  \end{equation}
  and
  \begin{equation}\label{ur}
  \|\phi\|_{H_{X,0}^1(\Omega)}^2\leq\left\{
               \begin{array}{ll}
               \ds  \left[\frac{-C_1J(\phi_0)(C_1-2)}{\|\phi_0\|_{H_{X,0}^1(\Omega)}^{C_1}}\right]^{\frac{2}{2-C_1}}(T-t)^{-\frac{2}{C_1-2}},&\hbox{~if~~}J(\phi_0)<0; \\
               \vspace{-0.05in}\\
               \ds  \left[\frac{C_1(C_1-2)(d-J(\phi_0))}{\|\phi_0\|_{H_{X,0}^1(\Omega)}^{C_1}}\right]^{\frac{2}{2-C_1}}(T-t)^{-\frac{2}{C_1-2}},&\hbox{~if~~}0\leq J(\phi_0)<d,
               \end{array}
             \right.
  \end{equation}
  where
  \begin{equation}\label{C1}
  C_1=\left\{
               \begin{array}{ll}
               \ds  p+1,&\hbox{ if~~}J(\phi_0)<0; \\
               \vspace{-0.05in}\\
               \ds  \frac{\left(\epsilon_0^{p+1}-1\right)(p-1)}{\epsilon_0^{p+1}}+2,&\hbox{ if~~}0\leq J(\phi_0)<d.
               \end{array}
             \right.
  \end{equation}
  Here
  \begin{equation}\label{68}
  \epsilon_0=\left[\frac{p+1}{2}-C_*^{\frac{2(p+1)}{p-1}}J(\phi_0)(p+1)\right]^{\frac{1}{p-1}}>1,
  \end{equation}
  and $C_*>0$ is given in \eqref{qianru}.
  \item{$(ii)$} the exponential growth estimate of blow-up solutions is given by
  \begin{equation*}
  \|\phi\|_{H_{X,0}^1(\Omega)}^2\geq\left\{
               \begin{array}{ll}
               \ds  \|\phi_0\|_{H_{X,0}^1(\Omega)}^2e^{C_2t},&\hbox{ if }J(\phi_0)<0; \\
               \vspace{-0.05in}\\
               \ds  \|\phi_0\|_{H_{X,0}^1(\Omega)}^2e^{C_3t},&\hbox{ if }0\leq J(\phi_0)<d,
               \end{array}
             \right.
  \end{equation*}
  where
  \begin{equation}\label{C23}
  C_2=\frac{\lambda_1(p-1)}{\lambda_1+1},\quad
  C_3=\frac{\lambda_1(p-1)\left(\epsilon_0^2-1\right)}{\epsilon_0^2(\lambda_1+1)}.
  \end{equation}
  Here $\lambda_1$ and $\epsilon_0$ are given in \eqref{tezhengzhi} and \eqref{68}, respectively.
\end{description}
\end{thm}

\begin{thm}\label{xiajie}
Suppose that $(A_1)$-$(A_3)$ hold. If the solution $\phi(t)$ to \eqref{1.1} blows up in finite time in the sense of $\lim\limits_{t\rightarrow T^-}\|\phi\|_{H_{X,0}^1(\Omega)}^2=\infty$, then the lower bound of blow-up rate for blow-up solutions is
\begin{equation}\label{88}
\|\phi\|_{H_{X,0}^1(\Omega)}^2\geq\frac{1}{C_*^{\frac{2(p+1)}{p-1}}(p-1)^{\frac{2}{p-1}}}(T-t)^{-\frac{2}{p-1}},
\end{equation}
where $C_*$ is defined by \eqref{qianru}.
\end{thm}

The following lemma provides some estimates of the solution, which will be used in the proof of Theorem \ref{baopo}.
\begin{lemma}\label{yinliqi}
Suppose that $(A_1)$-$(A_3)$ hold. Let $\phi=\phi(t)$, $t\in[0,T)$ be a weak solution to \eqref{1.1}. If $0\leq J(\phi_0)<d$ and $I(\phi_0)<0$, then there is a constant $\epsilon_2>\epsilon_1:=C_*^{-\frac{2}{p-1}}$ such that
\begin{equation}\label{9}
\|\phi\|_{p+1}\geq\epsilon_2,\quad t\in[0,T),
\end{equation}
and
\begin{equation}\label{11}
\frac{\epsilon_2}{\epsilon_1}\geq\epsilon_0:=\left[\frac{p+1}{2}-C_*^{\frac{2(p+1)}{p-1}}J(\phi_0)(p+1)\right]^{\frac{1}{p-1}}>1,
\end{equation}
where $C_*>0$ is given in \eqref{qianru}.
\end{lemma}
\begin{proof}
From \eqref{qianru}, $I(\phi_0)<0$ and \eqref{I}, we deduce that
\begin{equation*}
\frac{1}{C_*^2}\|\phi_0\|_{p+1}^2\leq\|X\phi_0\|_2^2<\|\phi_0\|_{p+1}^{p+1},
\end{equation*}
which gives
\begin{equation}\label{8}
\|\phi_0\|_{p+1}>C_*^{-\frac{2}{p-1}}=:\epsilon_1.
\end{equation}
By \eqref{J} and \eqref{qianru}, we obtain
\begin{equation}\label{7}
J(\phi)\geq\frac{1}{2C_*^2}\|\phi\|_{p+1}^2-\frac{1}{p+1}\|\phi\|_{p+1}^{p+1}.
\end{equation}
Let
\begin{equation*}
h(\epsilon):=\frac{1}{2C_*^2}\epsilon^2-\frac{1}{p+1}\epsilon^{p+1},\ \ \ \forall\epsilon\geq0.
\end{equation*}
Then $h(\epsilon)$ is increasing on $[0,\epsilon_1]$ and decreasing on $[\epsilon_1,\infty)$. By \eqref{djingquezhi}, we have
\begin{equation*}
h(\epsilon_1)=\sup\limits_{\epsilon\geq0}h(\epsilon)=d.
\end{equation*}
Noting that $J(\phi_0)<d$, there is a constant $\epsilon_2>\epsilon_1$ such that $h(\epsilon_2)=J(\phi_0)$. Let $\epsilon_*=\|\phi_0\|_{p+1}$. Then by \eqref{7}, we obtain $h(\epsilon_2)=J(\phi_0)\geq h(\epsilon_*)$, which, along with \eqref{8}, implies $\epsilon_2\leq\epsilon_*=\|\phi_0\|_{p+1}$.
We claim that $\|\phi\|_{p+1}\geq\epsilon_2$ for $t\in(0,T)$. Otherwise, there exists $t_0\in(0,T)$ such that $\|\phi(t_0)\|_{p+1}<\epsilon_2$. Since $\epsilon_1<\epsilon_2$, we can choose $t_0$ such that $\epsilon_1<\|\phi(t_0)\|_{p+1}<\epsilon_2$. Then from \eqref{7}, we deduce that
$$
J(\phi_0)=h(\epsilon_2)<h(\|\phi(t_0)\|_{p+1})\leq J(\phi(t_0)),
$$
which conflicts with \eqref{nengliangdengshi}. Thus, for $t\in[0,T)$, we have that $\|\phi\|_{p+1}\geq\epsilon_2$. Then by \eqref{qianru}, one has $C_*\|X\phi\|_2\geq\|\phi\|_{p+1}\geq\epsilon_2$. Hence, \eqref{9} holds.

Next, we prove \eqref{11}. Let $\epsilon_3=\epsilon_2/\epsilon_1>1$. Then by $h(\epsilon_2)=J(\phi_0)$ and $\epsilon_1=C_*^{-\frac{2}{p-1}}$, we obtain
\begin{equation*}
J(\phi_0)=h(\epsilon_1\epsilon_3)=\epsilon_1^2\epsilon_3^2\left(\frac{1}{2C_*^2}-\frac{\epsilon_1^{p-1}\epsilon_3^{p-1}}{p+1}\right)
=\frac{\epsilon_1^2\epsilon_3^2}{C_*^2}\left(\frac{1}{2}-\frac{\epsilon_3^{p-1}}{p+1}\right),
\end{equation*}
which, together with $J(\phi_0)\geq0$ and $\epsilon_3>1$, yields
\begin{equation}\label{10}
\frac{1}{2}-\frac{\epsilon_3^{p-1}}{p+1}=\frac{C_*^2J(\phi_0)}{\epsilon_1^2\epsilon_3^2}\leq\frac{C_*^2J(\phi_0)}{\epsilon_1^2}
=C_*^\frac{2(p+1)}{p-1}J(\phi_0).
\end{equation}
In accordance with \eqref{10}, \eqref{djingquezhi} and $J(\phi_0)<d$, we arrive at
\begin{equation*}
\frac{\epsilon_2}{\epsilon_1}=\epsilon_3\geq\left[\frac{p+1}{2}-C_*^{\frac{2(p+1)}{p-1}}J(\phi_0)(p+1)\right]^{\frac{1}{p-1}}
>\left(\frac{p+1}{2}-\frac{p-1}{2}\right)^{\frac{1}{p-1}}=1,
\end{equation*}
which gives \eqref{11}.
\end{proof}

\begin{proof}[Proof of Theorem \ref{baopo}]
The proof is divided into two parts.

{\bf Part 1: Upper bounds of blow-up time and rate.}

For $t\in[0,T)$, we define
\begin{equation}\label{18}
M(t):=\frac{1}{2}\|\phi\|_{H_{X,0}^1(\Omega)}^2
\end{equation}
and
\begin{equation}\label{17}
F(t):=\left\{
               \begin{array}{ll}
               \ds  -J(\phi),&\hbox{ if }J(\phi_0)<0; \\
               \vspace{-0.05in}\\
               \ds  d-J(\phi),&\hbox{ if }0\leq J(\phi_0)<d.
               \end{array}
             \right.
\end{equation}

When $J(\phi_0)<0$, it follows from \eqref{17}, \eqref{nengliangdengshi} and \eqref{J} that
\begin{equation}\label{15}
0<F(0)\leq F(t)=-\frac{1}{2}\|X\phi\|_2^2+\frac{1}{p+1}\|\phi\|_{p+1}^{p+1}\leq\frac{1}{p+1}\|\phi\|_{p+1}^{p+1}.
\end{equation}
According to \eqref{18}, \eqref{Idengshi}, \eqref{J}, \eqref{I}, \eqref{17} and \eqref{15}, we arrive at
\begin{equation}\label{19}
M'(t)=-I(\phi)=\frac{p-1}{p+1}\|\phi\|_{p+1}^{p+1}-2J(\phi)\geq(p+1)F(t).
\end{equation}
When $0\leq J(\phi_0)<d$, we deduce from \eqref{17}, \eqref{nengliangdengshi}, \eqref{J}, \eqref{djingquezhi}, \eqref{qianru} and Lemma \ref{yinliqi} that
\begin{equation}\label{16}
\begin{split}
0<F(0)\leq F(t)&=d-\frac{1}{2}\|X\phi\|_2^2+\frac{1}{p+1}\|\phi\|_{p+1}^{p+1}\\
&\leq d-\frac{1}{2C_*^2}\|\phi\|_{p+1}^2+\frac{1}{p+1}\|\phi\|_{p+1}^{p+1}\\
&\leq d-\frac{\epsilon_1^2}{2C_*^2}+\frac{1}{p+1}\|\phi\|_{p+1}^{p+1}\\
&=\frac{p-1}{2(p+1)}C_*^{-\frac{2(p+1)}{p-1}}-\frac{1}{2}C_*^{-\frac{2(p+1)}{p-1}}+\frac{1}{p+1}\|\phi\|_{p+1}^{p+1}\\
&\leq\frac{1}{p+1}\|\phi\|_{p+1}^{p+1}.
\end{split}\end{equation}
From \eqref{18}, \eqref{Idengshi}, \eqref{J}, \eqref{I}, \eqref{17}, \eqref{djingquezhi}, Lemma \ref{yinliqi} and \eqref{16}, we obtain that
\begin{equation}\label{12}
\begin{split}
M'(t)=-I(\phi)&=\frac{p-1}{p+1}\|\phi\|_{p+1}^{p+1}-2J(\phi)\\
&=\frac{p-1}{p+1}\|\phi\|_{p+1}^{p+1}+2F(t)-2d\\
&=\frac{p-1}{p+1}\|\phi\|_{p+1}^{p+1}+2F(t)-\frac{p-1}{p+1}C_*^{-\frac{2(p+1)}{p-1}}\\
&=\frac{p-1}{p+1}\|\phi\|_{p+1}^{p+1}+2F(t)-\frac{p-1}{p+1}\left(\frac{\epsilon_1}{\epsilon_2}\right)^{p+1}\epsilon_2^{p+1}\\
&\geq\frac{p-1}{p+1}\|\phi\|_{p+1}^{p+1}+2F(t)-\frac{p-1}{\epsilon_0^{p+1}(p+1)}\|\phi\|_{p+1}^{p+1}\\
&\geq \left[\frac{\left(\epsilon_0^{p+1}-1\right)(p-1)}{\epsilon_0^{p+1}}+2\right]F(t).
\end{split}\end{equation}
Then by \eqref{15}, \eqref{19}, \eqref{16} and \eqref{12}, we arrive at
\begin{equation}\label{20}
M'(t)\geq C_1F(t)>0,
\end{equation}
where $C_1$ is defined in \eqref{C1}.
Noting  \eqref{nengliangdengshi}, we have
\begin{equation*}
F'(t)=\|\phi_t\|_{H_{X,0}^1(\Omega)}^2,
\end{equation*}
which, together with \eqref{18}, \eqref{Hfanshu}, H\"{o}lder's inequality and \eqref{20}, implies
\begin{equation*}\begin{split}
M(t)F'(t)&=\frac{1}{2}\|\phi\|_{H_{X,0}^1(\Omega)}^2\|\phi_t\|_{H_{X,0}^1(\Omega)}^2\\
&\geq\frac{1}{2}\left(\|X\phi\|_2^2\|X\phi_t\|_2^2+2\|\phi\|_2\|\phi_t\|_2\|X\phi\|_2\|X\phi_t\|_2
               +\|\phi\|_2^2\|\phi_t\|_2^2\right)\\
&\geq\frac{1}{2}\left[(X\phi,X\phi_t)+(\phi,\phi_t)\right]^2\\
&=\frac{1}{2}(M'(t))^2\geq\frac{C_1}{2}M'(t)F(t).
\end{split}\end{equation*}
Therefore
\begin{equation*}
\frac{F'(t)}{F(t)}\geq\frac{C_1M'(t)}{2M(t)}.
\end{equation*}
Integrating it from $0$ to $t$, we obtain
\begin{equation*}
\frac{F(t)}{(M(t))^{\frac{C_1}{2}}}\geq\frac{F(0)}{(M(0))^{\frac{C_1}{2}}},
\end{equation*}
which, along with \eqref{20}, yields
\begin{equation}\label{22}
\frac{M'(t)}{(M(t))^{\frac{C_1}{2}}}\geq\frac{C_1F(0)}{(M(0))^{\frac{C_1}{2}}}.
\end{equation}
Integrating it from $0$ to $t$, one has
\begin{equation}\label{23}
(M(t))^{\frac{2-C_1}{2}}\leq(M(0))^{\frac{2-C_1}{2}}-\frac{C_1F(0)(C_1-2)}{2(M(0))^{\frac{C_1}{2}}}t.
\end{equation}
Since $C_1>2$, \eqref{23} cannot hold for all $t\geq0$, which implies that $M(t)$ blows up at some finite time $T$, namely,
\begin{equation}\label{25}
\lim_{t\rightarrow T^-}M(t)=\infty,
\end{equation}
and \eqref{ut} holds.
In addition, integrating \eqref{22} from $t$ to $T$ and using \eqref{25}, we deduce that
\begin{equation*}
M(t)\leq\left[\frac{C_1F(0)(C_1-2)}{2(M(0))^{\frac{C_1}{2}}}\right]^{\frac{2}{2-C_1}}(T-t)^{-\frac{2}{C_1-2}},
\end{equation*}
which gives \eqref{ur}.

{\bf Part 2: Exponential growth estimate.}

When $J(\phi_0)<0$, in accordance with \eqref{18}, \eqref{Idengshi}, \eqref{JI}, \eqref{nengliangdengshi}, \eqref{Hfanshu} and \eqref{tezhengzhi}, we have
\begin{equation*}\begin{split}
M'(t)=-I(\phi)&\geq\frac{p-1}{2}\|X\phi\|_2^2-(p+1)J(\phi_0)\\
&>\frac{p-1}{2}\|X\phi\|_2^2\\
&\geq\frac{\lambda_1(p-1)}{\lambda_1+1}M(t),
\end{split}\end{equation*}
which, along with \eqref{18}, yields
\begin{equation*}
\|\phi\|_{H_{X,0}^1(\Omega)}^2\geq\|\phi_0\|_{H_{X,0}^1(\Omega)}^2e^{C_2t},
\end{equation*}
where $C_2$ is defined in \eqref{C23}.

When $0\leq J(\phi_0)<d$, it follows from \eqref{18}, \eqref{Idengshi}, \eqref{JI}, \eqref{nengliangdengshi}, \eqref{djingquezhi}, Lemma \ref{yinliqi}, \eqref{Hfanshu} and \eqref{tezhengzhi} that
\begin{equation*}\begin{split}
M'(t)=-I(\phi)&\geq\frac{p-1}{2}\|X\phi\|_2^2-(p+1)J(\phi_0)\\
&>\frac{p-1}{2}\|X\phi\|_2^2-\frac{p-1}{2}C_*^{-\frac{2(p+1)}{p-1}}\\
&=\frac{p-1}{2}\left(\frac{\epsilon_1^2}{\epsilon_2^2}\|X\phi\|_2^2+\frac{\epsilon_2^2-\epsilon_1^2}{\epsilon_2^2}\|X\phi\|_2^2\right)
-\frac{p-1}{2C_*^2}\epsilon_1^2\\
&\geq\frac{(p-1)\left(\epsilon_0^2-1\right)}{2\epsilon_0^2}\|X\phi\|_2^2\\
&\geq\frac{\lambda_1(p-1)\left(\epsilon_0^2-1\right)}{\epsilon_0^2(\lambda_1+1)}M(t),
\end{split}\end{equation*}
which, together with \eqref{18}, implies
\begin{equation*}
\|\phi\|_{H_{X,0}^1(\Omega)}^2\geq\|\phi_0\|_{H_{X,0}^1(\Omega)}^2e^{C_3t},
\end{equation*}
where $C_3$ is defined in \eqref{C23}.
\end{proof}


\begin{proof}[Proof of Theorem \ref{xiajie}]
Define
\begin{equation}\label{18-}
M(t):=\frac{1}{2}\|\phi\|_{H_{X,0}^1(\Omega)}^2,\ \ \ t\in[0,T).
\end{equation}
Then by $\lim\limits_{t\rightarrow T^-}\|\phi\|_{H_{X,0}^1(\Omega)}^2=\infty$, we have
\begin{equation}\label{25-}
\lim_{t\rightarrow T^-}M(t)=\infty.
\end{equation}
By \eqref{18-}, \eqref{Idengshi}, \eqref{I}, \eqref{qianru} and \eqref{Hfanshu}, one can conclude that
\begin{equation}\label{26}
M'(t)=-\|X\phi\|_2^2+\|\phi\|_{p+1}^{p+1}\leq C_*^{p+1}\|X\phi\|_2^{p+1}\leq 2^{\frac{p+1}{2}}C_*^{p+1}(M(t))^{\frac{p+1}{2}},
\end{equation}
which gives
\begin{equation*}
\frac{M'(t)}{(M(t))^{\frac{p+1}{2}}}\leq2^{\frac{p+1}{2}}C_*^{p+1}.
\end{equation*}
Integrating it from $t$ to $T$ and using \eqref{25-}, we arrive at
\begin{equation*}
M(t)\geq\frac{1}{2C_*^{\frac{2(p+1)}{p-1}}(p-1)^{\frac{2}{p-1}}}(T-t)^{-\frac{2}{p-1}},
\end{equation*}
which, along with \eqref{18-}, implies \eqref{88} .
\end{proof}

\subsection{Asymptotic convergence of the global solution}
\quad\quad
This subsection establishes that the global solution of \eqref{1.1} strongly converge to the solution of \eqref{x18-11} as time approaches infinity, which answers question $(Q_3)$.

Before stating the result, we need to introduce the stationary problem corresponding to \eqref{1.1}, i.e., the following boundary value problem
\begin{align}\label{x18-11}
 \left\{\begin{array}{ll}
     \ds -\Delta_X\phi=|\phi|^{p-1}\phi,\quad &x\in\Omega,\\
     \ds \phi(x)=0,\quad &x\in\partial\Omega.
    \end{array}\right.
\end{align}
According to Theorem 1.1 in \cite{MR4354732,MR4491354}, one has that the problem \eqref{x18-11} admits multiple nontrivial weak solutions in $H_{X,0}^1(\Omega)$. We call $\phi\in H_{X,0}^1(\Omega)$ is a solution of \eqref{x18-11} if
\begin{equation}\label{x19-11}
\langle J'(\phi),\nu\rangle=(X\phi,X\nu)-\left(|\phi|^{p-1}\phi,\nu\right)=0
\end{equation}
for any $\nu\in H_{X,0}^1(\Omega)$, where $J'(\phi)$ is the Fr\'{e}chet derivative of $J(\phi)$, $\langle\cdot,\cdot\rangle$ represents the dual product between $H_{X,0}^1(\Omega)$ and $H_X^{-1}(\Omega)$, and $H_X^{-1}(\Omega)$ denotes the dual space of $H_{X,0}^1(\Omega)$. Define $\Psi$ as the set of all solutions to problem \eqref{x18-11}, i.e.,
\begin{equation*}
\Psi=\{\psi\in H_{X,0}^1(\Omega):\langle J'(\psi),\nu\rangle=0,\forall\nu\in H_{X,0}^1(\Omega)\}.
\end{equation*}

\begin{thm}\label{xx-33}
Suppose that $(A_1)$-$(A_3)$ hold. Let $\phi=\phi(t)$ be a global solution to \eqref{1.1}. Then there exists $\phi^*\in \Psi$ and an increasing sequence $\{t_k\}_{k=1}^\infty$ with $t_k\rightarrow\infty$ as $k\rightarrow\infty$ such that
\begin{equation*}
\lim_{k\rightarrow\infty}\|\phi(t_k)-\phi^*\|_{H_{X,0}^1(\Omega)}=0.
\end{equation*}
\end{thm}

\begin{proof}
Let $\phi=\phi(t)$ be the global solution to problem \eqref{1.1}.
First, we prove
\begin{equation}\label{x1-33}
0\leq J(\phi)\leq J(\phi_0),\quad \forall t\in[0,\infty).
\end{equation}
From \eqref{nengliangdengshi}, it follows that $J(\phi(t))$ is nonincreasing with respect to $t$ and $J(\phi)\leq J(\phi_0)$ holds for all $t\in[0,\infty)$. We claim that $J(\phi)\geq 0$ holds for all $t\in[0,\infty)$. If not, there would exist $t_0\in[0,\infty)$ such that $J(\phi(t_0))<0$. Combining this with \eqref{JI}, we have $I(\phi(t_0))<0$. According to Theorem \ref{baopo}, $\phi$ would blow up in finite time, which contradicts the fact that $\phi$ is a global solution to problem \eqref{1.1}. Therefore, \eqref{x1-33} holds.

From \eqref{x1-33} and the fact that $J(\phi(t))$ is non-increasing with respect to $t$ on $[0,\infty)$, it follows that there exists a positive constant $C_0\in[0,J(\phi_0)]$ such that
\begin{equation}\label{x2-33}
\lim\limits_{t\rightarrow\infty}J(\phi)=C_0.
\end{equation}
Letting $t\rightarrow\infty$ in \eqref{nengliangdengshi}, we obtain
\begin{equation}\label{x3-33}
\int_0^\infty\|\phi'(\tau)\|_{H_{X,0}^1(\Omega)}^2d\tau=J(\phi_0)-C_0\leq J(\phi_0),
\end{equation}
which implies the existence of an increasing sequence $\{t_k\}_{k=1}^\infty$ with $t_k\rightarrow\infty)$ as $k\rightarrow\infty$ such that
\begin{equation}\label{x4-33}
\lim_{k\rightarrow\infty}\|\phi'(t_k)\|_{H_{X,0}^1(\Omega)}=0.
\end{equation}
According to \eqref{J} and \eqref{weakdengshi}, for any $\varphi\in H_{X,0}^1(\Omega)$, we have
\begin{equation}\label{x5-33}
\begin{split}
\langle J'(\phi(t)),\varphi\rangle&=\frac{d}{d\tau}J(\phi(t)+\tau\varphi)\bigg|_{\tau=0}\\
&=\left[\left(X(\phi(t)+\tau\varphi),X\varphi\right)-\left(|\phi(t)+\tau\varphi|^{p-1}(\phi(t)+\tau\varphi),\varphi\right)\right]\big|_{\tau=0}\\
&=(X\phi(t),X\varphi)-\left(|\phi(t)|^{p-1}\phi(t),\varphi\right)\\
&=-(\phi'(t),\varphi)-(X\phi'(t),X\varphi).
\end{split}
\end{equation}
Using \eqref{x5-33}, \eqref{Hfanshu} and \eqref{tezhengzhi}, we deduce
\begin{equation*}
\begin{split}
\|J'(\phi(t_k))\|_{H_X^{-1}(\Omega)}&=\sup_{\|\varphi\|_{H_{X,0}^1(\Omega)}\leq1}\big|\langle J'(\phi(t_k)),\varphi\rangle\big|\\
&\leq\sup_{\|\varphi\|_{H_{X,0}^1(\Omega)}\leq1}\left(\|\phi'(t_k)\|_2\|\varphi\|_2+\|X\phi'(t_k)\|_2\|X\varphi\|_2\right)\\
&\leq\sup_{\|\varphi\|_{H_{X,0}^1(\Omega)}\leq1}\|X\varphi\|_2\left(\lambda_1^{-\frac{1}{2}}\|\phi'(t_k)\|_2+\|X\phi'(t_k)\|_2\right)\\
&\leq\left(\lambda_1^{-\frac{1}{2}}+1\right)\|\phi'(t_k)\|_{H_{X,0}^1(\Omega)},
\end{split}
\end{equation*}
which, combined with \eqref{x4-33}, yields
\begin{equation}\label{x6-33}
\lim_{k\rightarrow\infty}\|J'(\phi(t_k))\|_{H_X^{-1}(\Omega)}=0.
\end{equation}
So there exists a constant $c>0$ such that
\begin{equation*}
\begin{split}
|I(\phi(t_k))|&=|\langle J'(\phi(t_k)),\phi(t_k)\rangle|\\
&\leq\|J'(\phi(t_k))\|_{H_X^{-1}(\Omega)}\|\phi(t_k)\|_{H_{X,0}^1(\Omega)}\\
&\leq c\|\phi(t_k\|_{H_{X,0}^1(\Omega)}.
\end{split}
\end{equation*}
It follows from \eqref{x4-33}, \eqref{Hfanshu} and \eqref{tezhengzhi} that
\begin{equation*}
\begin{split}
J(\phi_0)+\frac{c}{p+1}\|\phi(t_k)\|_{H_{X,0}^1(\Omega)}&\geq J(\phi(t_k))-\frac{1}{p+1}I(\phi(t_k))\\
&=\frac{p-1}{2(p+1)}\|X\phi(t_k)\|_2^2\\
&\geq\frac{(p-1)\lambda_1}{2(p+1)(1+\lambda_1)}\|\phi(t_k)\|_{H_{X,0}^1(\Omega)}^2.
\end{split}
\end{equation*}
Thus, there exists a constant $K_1>0$ such that
\begin{equation}\label{x8-33}
\|\phi(t_k)\|_{H_{X,0}^1(\Omega)}\leq K_1, \quad k=1,2,\cdots.
\end{equation}
Since the embedding $H_{X,0}^1(\Omega)\hookrightarrow L^{p+1}(\Omega)$ is compact, there exists an increasing subsequence of $\{t_k\}_{k=1}^\infty$ (still denoted as $\{t_k\}_{k=1}^\infty$) and $u_*\in H_{X,0}^1(\Omega)$ such that, as $k\rightarrow\infty$ (denote $\phi_k:=\phi(t_k)$):
\begin{align}
&\phi_k\rightharpoonup \phi_* \hbox{ weakly in } H_{X,0}^1(\Omega),\label{x9-33}\\
&\phi_k\rightarrow \phi_* \hbox{ strongly in } L^{p+1}(\Omega).\label{x10-33}
\end{align}
By \eqref{J}, we have
\begin{equation*}
\begin{split}
\quad\langle J'(\phi_k),\phi_k-\phi_*\rangle&=\frac{d}{d\tau}J(\phi_k+\tau (\phi_k-\phi_*))\bigg|_{\tau=0}\\
&=\big\{\left(X[\phi_k+\tau(\phi_k-\phi_*)],X(\phi_k-\phi_*)\right)\\
&\quad-\left(|\phi_k+\tau(\phi_k-\phi_*)|^{p-1}[\phi_k+\tau(\phi_k-\phi_*)],\phi_k-\phi_*\right)\big\}\big|_{\tau=0}\\
&=\left(X\phi_k,X(\phi_k-\phi_*)\right)-\left(|\phi_k|^{p-1}\phi_k,\phi_k-\phi_*\right).
\end{split}
\end{equation*}
Similarly,
\begin{equation}\label{x12-33}
\langle J'(\phi_*),\phi_k-u_*\rangle=(X\phi_*,X(\phi_k-\phi_*))-\left(|\phi_*|^{p-1}\phi_*,\phi_k-\phi_*\right).
\end{equation}
Thus, we obtain
\begin{equation}\label{x16-33}
\quad \langle J'(\phi_k)-J'(\phi_*),\phi_k-\phi_*\rangle=\|X(\phi_k-\phi_*)\|_2^2-\Theta,
\end{equation}
where $$\Theta:=\left(|\phi_k|^{p-1}\phi_k-|\phi_*|^{p-1}\phi_*,\phi_k-\phi_*\right).$$
Then by H\"{o}lder inequality, \eqref{x1111-33} and \eqref{x8-33}, one has
\begin{equation*}
\begin{split}
|\Theta|&\leq\|\phi_*^p-\phi_k^p\|_{\frac{p+1}{p}}\|\phi_k-\phi_*\|_{p+1}\\
&\leq\left(\|\phi_*\|_{p+1}^p+\|\phi_k\|_{p+1}^p\right)\|\phi_k-\phi_*\|_{p+1}\\
&\leq\left(\|\phi_*\|_{p+1}^p+C^p\|\phi_k\|_{H_{X,0}^1(\Omega)}^p\right)\|\phi_k-\phi_*\|_{p+1}\\
&\leq\left(\|\phi_*\|_{p+1}^p+C^pK_1^p\right)\|\phi_k-\phi_*\|_{p+1},
\end{split}
\end{equation*}
which, combined with \eqref{x10-33}, implies
\begin{equation}\label{x11-33}
\lim_{k\rightarrow\infty}|\Theta|=0.
\end{equation}
From \eqref{x9-33}, \eqref{x10-33} and \eqref{x12-33}, we obtain
\begin{equation}\label{x13-33}
\lim_{k\rightarrow\infty}\langle J'(\phi_*),\phi_k-\phi_*\rangle=0.
\end{equation}
By \eqref{x8-33}, we have
\begin{equation*}
\begin{split}
\big|\langle J'(\phi_k),\phi_k-u_*\rangle\big|&\leq\|J'(\phi_k)\|_{H_X^{-1}(\Omega)}\left(\|\phi_k\|_{H_{X,0}^1(\Omega)}+\|\phi_*\|_{H_{X,0}^1(\Omega)}\right)\\
&\leq\|J'(\phi_k)\|_{H_X^{-1}(\Omega)}\left(K_1+\|\phi_*\|_{H_{X,0}^1(\Omega)}\right),
\end{split}
\end{equation*}
and combining this with \eqref{x6-33}, we arrive at
\begin{equation}\label{x15-33}
\lim_{k\rightarrow\infty}\big|\langle J'(\phi_k),\phi_k-\phi_*\rangle\big|=0.
\end{equation}
It follows from \eqref{x16-33}, \eqref{x11-33}, \eqref{x13-33} and \eqref{x15-33} that
\begin{equation*}
\lim_{k\rightarrow\infty}\|X(\phi_k-\phi_*)\|_2^2=\lim_{k\rightarrow\infty}\left[\langle J'(\phi_k)-J'(\phi_*),\phi_k-\phi_*\rangle+\Theta\right]=0.
\end{equation*}
By \eqref{Hfanshu} and \eqref{tezhengzhi}, we conclude
\begin{equation*}
\lim_{k\rightarrow\infty}\|\phi_k-\phi_*\|_{H_{X,0}^1(\Omega)}^2=0.
\end{equation*}

Finally, we prove $\phi_*\in \Psi$. From \eqref{x10-33}, there exists a subsequence of $\{\phi_k\}_{k=1}^\infty$ (still denoted as $\{\phi_k\}_{k=1}^\infty$) such that, as $k\rightarrow\infty$,
\begin{equation}\label{x21-33}
\phi_k(x)\rightarrow \phi_*(x) \hbox{ a.e. in } \Omega,
\end{equation}
and there exists $w\in L^{p+1}(\Omega)$ satisfying, for all $k$,
\begin{equation}\label{x22-33}
|\phi_k(x)|\leq w(x) \hbox{ a.e. in } \Omega.
\end{equation}
By \eqref{x21-33}, we conclude that, as $k\rightarrow\infty$,
\begin{equation}\label{x24-33}
|\phi_k(x)|^{p-1}\phi_k(x)\nu(x)\rightarrow |\phi_*(x)|^{p-1}\phi_*(x)\nu(x) \hbox{ a.e. in } \Omega
\end{equation}
for $\nu\in H_{X,0}^1(\Omega)$, and it follows from \eqref{x22-33} that, for all $k$,
\begin{equation}\label{x25-33}
\left||\phi_k(x)|^{p-1}\phi_k(x)\nu(x)\right|\leq \left||w(x)|^{p-1}w(x)\nu(x)\right| \hbox{ a.e. in } \Omega.
\end{equation}
Applying H\"{o}lder inequality and \eqref{x1111-33}, we obtain
\begin{equation*}
\int_\Omega\left||w(x)|^{p-1}w(x)\nu(x)\right|dx\leq\left\||w(x)|^p\right\|_{\frac{p+1}{p}}\|\nu(x)\|_{p+1}\leq C\|w(x)\|_{p+1}^p\|\nu(x)\|_{H_{X,0}^1(\Omega)}.
\end{equation*}
Combining $w\in L^{p+1}(\Omega)$ and $\nu\in H_{X,0}^1(\Omega)$, one has
\begin{equation*}
\int_\Omega\left||w(x)|^{p-1}w(x)\nu(x)\right|dx<\infty.
\end{equation*}
Hence,
\begin{equation}\label{vvv}
|w(x)|^{p-1}w(x)\nu(x)\in L^1(\Omega).
\end{equation}
By the Lebesgue dominated convergence theorem, together with \eqref{x24-33}, \eqref{x25-33} and \eqref{vvv}, it follows that, for any $\nu\in H_{X,0}^1(\Omega)$,
\begin{equation}\label{x23-33}
\left(|\phi_k|^{p-1}\phi_k,\nu\right)\rightarrow\left(|\phi_*|^{p-1}\phi_*,\nu\right) \mbox{  as  } k\rightarrow\infty.
\end{equation}
Since
\begin{equation*}
\langle J'(\phi_k),\nu\rangle=(X\phi_k,X\nu)-\left(|\phi_k|^{p-1}\phi_k,\nu\right),
\end{equation*}
letting $k\rightarrow \infty$ and using \eqref{x6-33}, \eqref{x9-33} and \eqref{x23-33}, we obtain
\begin{equation*}
0=(X\phi_*,X\nu)-\left(|\phi_*|^{p-1}\phi_*,\nu\right).
\end{equation*}
Therefore, $\phi_*\in \Psi$.
\end{proof}

\begin{remark}
Following Theorem 1.4 in \cite{MR3961341}, we note that the global solution decays exponentially to zero as $t\rightarrow\infty$ when the solution satisfies some initial conditions, while our Theorem \ref{xx-33} demonstrates that any global solution of \eqref{1.1} strongly converge to the solution of the stationary problem \eqref{x18-11} as $t\rightarrow \infty$, which improves upon the weak convergence result in Theorem 1.5 of \cite{MR4854159}.
\end{remark}

\section*{Declarations}
{\bf Funding:} this work was supported by the National Natural Science Foundation of China (12171339).

\noindent {\bf Ethical Approval:} not applicable.

\noindent {\bf Informed Consent:} not applicable.

\noindent {\bf Author Contributions:} each of the authors contributed to each part of this study equally. All authors read and proved the final vision of the manuscript.

\noindent {\bf Data Availability Statement:} not applicable. 

\noindent {\bf Conflict of Interest:} the authors declare that they have no conflict of interest.

\noindent {\bf Clinical Trial Number:} not applicable.



\end{document}